\documentclass[12pt,a4paper,reqno]{amsart}

\RequirePackage{iftex}
\ifPDFTeX
\usepackage[T1]{fontenc}
\usepackage[utf8]{inputenc}
\usepackage{textcomp}
\usepackage{fix-cm}
\fi

\usepackage{amssymb,amsthm,mathtools,bm}
\IfFileExists{newtxtext.sty}{%
	\IfFileExists{newtxmath.sty}{\usepackage{newtxtext,newtxmath}}{\usepackage{mathptmx}}%
}{\usepackage{mathptmx}}
\usepackage{cite}
\usepackage{hyperref}
\usepackage{booktabs}
\usepackage{array}
\usepackage{chngcntr}
\usepackage{enumitem}
\usepackage{xcolor}
\allowdisplaybreaks[4]
\hypersetup{hidelinks}

\evensidemargin=\oddsidemargin
\advance\voffset-20mm
\newtheorem{theorem}{Theorem}[section]
\newtheorem{lemma}[theorem]{Lemma}
\newtheorem{proposition}[theorem]{Proposition}
\newtheorem{corollary}[theorem]{Corollary}
\theoremstyle{definition}
\newtheorem{definition}[theorem]{Definition}
\newtheorem{remark}[theorem]{Remark}
\counterwithout{equation}{section}

\DeclareMathOperator{\ev}{ev}
\DeclareMathOperator{\rank}{rank}
\DeclareMathOperator{\Hull}{Hull}
\DeclareMathOperator{\Span}{Span}

\newcommand{\Fq}{\mathbb F_q}
\newcommand{\PRM}{\mathrm{PRM}}
\newcommand{\PP}{\mathbb P}
\newcommand{\Qm}{Q}
\newcommand{\ActSet}{\mathcal A}
\newcommand{\TopSet}{\mathcal T}
\newcommand{\RemSet}{\mathcal R}
\newcommand{\Lift}{\Lambda}
\newcommand{\Red}{\operatorname{red}}

\title[Recursive Structure of Hulls of PRM Codes]{Recursive Structure of Hulls of Projective Reed-Muller Codes}

\author[Y.\ Song]{Yufeng Song}
\thanks{Yufeng Song is with the Department of Mathematics, Nanjing University of Aeronautics and Astronautics, Nanjing, China (e-mail: yufengsong.math@qq.com).}

\author[Q.\ Yue]{Qin Yue}
\thanks{Qin Yue is with the Department of Mathematics, Nanjing University of Aeronautics and Astronautics, Nanjing, China, and with the State Key Laboratory of Cryptology, Beijing, China (e-mail: yueqin@nuaa.edu.cn).}

\date{}

\begin{document}
	
	\begin{abstract}
		For a nonnegative integer $r$ and a positive integer $v$ satisfying
		\[
		\frac{r(q-1)}{2}<v<\frac{(r+1)(q-1)}{2},
		\]
		we define the combinatorial numbers
		\[
		A_r(v)=
		\begin{cases}
			\displaystyle
			\sum_{t=r(q-1)-v}^{v}\ \sum_{j=0}^{r}(-1)^j\binom{r}{j}\binom{t-jq+r-1}{r-1}, & r>0,\\[1.2ex]
			1, & r=0.
		\end{cases}
		\]
		For the projective Reed-Muller code $\PRM(q,m,v)$, we determine its hull dimension:
		\[
		\dim \Hull\bigl(\PRM(q,m,v)\bigr)
		=
		\dim \PRM(q,m,v)
		-
		\sum_{i=0}^{\ell}A_{2i+\epsilon}\bigl(v-(\ell-i)(q-1)\bigr),
		\]
		where
		\[
		\ell=\Bigl\lfloor\frac r2\Bigr\rfloor,\qquad
		\epsilon=
		\begin{cases}
			0, & r\ \text{is even},\\
			1, & r\ \text{is odd}.
		\end{cases}
		\]
		This formula applies in the open lower-half range 
		$
		0<v<\frac{m\Qm}{2},
		$
		equivalently for $v\in I_r$ with $m\ge r+1$; the range
		$
		\frac{m\Qm}{2}<v<m\Qm
		$
		is then obtained by S\o rensen's duality theorem \cite{Sorensen}.
	\end{abstract}
	
	
	\maketitle

\section{Introduction}

Projective Reed-Muller codes are projective analogues of generalized Reed-Muller codes.

Kasami, Lin, and Peterson \cite{KasamiLinPetersonPrimitive} and Weldon \cite{Weldon} independently introduced generalized Reed-Muller codes. Kasami, Lin, and Peterson \cite{KasamiLinPeterson} also introduced polynomial codes, while Delsarte, Goethals, and MacWilliams \cite{DelsarteGoethalsMacWilliams} studied the relation between the multivariable and one-variable descriptions.

Manin and Vladut \cite{ManinVladut} first mentioned these codes in the context of algebraic geometry codes. Lachaud \cite{Lachaud88} introduced the name projective Reed-Muller codes and studied the cases $v=1$ and $v=2$. Lachaud \cite{Lachaud90} determined the parameters when $1\le v<q$. S\o rensen \cite{Sorensen} determined the parameters for $1\le v\le m(q-1)$, described the dual codes, and studied the cyclic case. Berger \cite{Berger02}, Kaplan \cite{KaplanWeights17}, Beelen, Datta, and Ghorpade \cite{BeelenDattaGhorpade18,BeelenDattaGhorpade22}, and Ghorpade and Ludhani \cite{GhorpadeLudhani24} studied further properties of these codes. More recently, Gimenez, Ruano, and San-Jos\'e \cite{GimenezRuanoSanJose24} studied subfield subcodes of projective Reed-Muller codes. San-Jos\'e \cite{SanJoseRecursive24} gave a recursive description of projective Reed-Muller codes, and San-Jos\'e \cite{SanJoseDecoding26} studied recursive decoding for these codes. Nardi and San-Jos\'e \cite{NardiSanJose26} initiated the weighted projective Reed-Muller side.

In this paper, we consider the hulls of projective Reed-Muller codes. The hull of a linear code $C$ is
\[
\Hull(C)=C\cap C^\perp.
\]
Hulls have been studied for many families of codes, because of their applications to quantum error-correcting codes \cite{GuendaJitmanGulliver,GaoYueHuangZhang21,ChenLingLiu23}. For projective Reed-Muller codes over the projective plane, Ruano and San-Jos\'e determined the hulls in \cite{RuanoSanJose24}; see also \cite{RuanoSanJose25}. In arbitrary projective dimension, Kaplan and Kim and, later, Song and Luo determined several low-interval and boundary cases for PRM hulls \cite{KaplanKim,SongLuo}.

Throughout the paper we denote by $\Qm$ the integer $q-1$ and set
\[
I_r=\Bigl(\frac{r\Qm}{2},\frac{(r+1)\Qm}{2}\Bigr)\qquad (r\ge 0).
\]
Since S\o rensen described the dual codes, the intervals above the midpoint can be reduced to lower-half intervals when $v\not\equiv0\pmod{\Qm}$; the remaining upper-half case $v\equiv0\pmod{\Qm}$ was determined by Song and Luo \cite{SongLuo}. Therefore, for a fixed interval $I_r$, the condition that this interval lies in the open lower half is
\[
I_r\subset \Bigl(0,\frac{m\Qm}{2}\Bigr)
\qquad\Longleftrightarrow\qquad
m\ge r+1.
\]
In particular, $m=r+1$ is the least ambient dimension for which $I_r$ lies in the open lower half.

For $v\in I_r$ and $m\ge r+1$, put
\[
\Delta_r(v)=\dim \PRM(q,m,v)-\dim \Hull\bigl(\PRM(q,m,v)\bigr).
\]
Prior work provides the low-interval and boundary inputs needed to anchor the general picture \cite{KaplanKim,SongLuo}. We prove a recursion from $I_{r-2}$ to $I_r$ for $r\ge 2$. The intervals above the midpoint are then obtained by S\o rensen's duality together with the Song and Luo case \cite{SongLuo} for
\[
\frac{m(q-1)}{2}<v<m(q-1),\qquad v\equiv0\pmod{q-1}.
\]

We prove that if $v\in I_r$ with $r\ge 2$, then
\[
\Delta_r(v)=A_r(v)+\Delta_{r-2}(v-(q-1)),
\]
	where $A_r(v)$ is an explicit combinatorial number. When $m\ge r+1$, equivalently when $I_r\subset (0,m\Qm/2)$, this recursion determines the hull defect in the open lower-half interval $I_r$. The interval $(m\Qm/2,m\Qm)$ is then obtained by S\o rensen's duality together with the Song and Luo case \cite{SongLuo}. The proof is based on a structured Gram-matrix analysis and also yields full-rank principal blocks.

The paper is organized as follows. In Section~II we recall the definition of projective Reed-Muller codes, the hull-rank identity, and the Gram-matrix model, and we collect the elementary facts used below. Sections~III--V develop the recursive step and its supporting constructions. Section~VI gives the explicit formulas and the resulting hull dimensions. Section~VII is the conclusion.

\section{Preliminaries}

\subsection{Projective Reed-Muller codes}

Projective Reed-Muller codes are evaluation codes over projective space. Let the projective space over $\Fq$ be
\[
\PP^m(\Fq)=\bigl(\Fq^{m+1}\setminus\{0\}\bigr)/\sim,
\]
where
\[
(a_0,\dots,a_m)\sim (b_0,\dots,b_m)
\]
if and only if $(b_0,\dots,b_m)=\lambda(a_0,\dots,a_m)$ for some $\lambda\in\Fq^*$. The number of projective points is
\[
n=\frac{q^{m+1}-1}{q-1}.
\]
For every projective point we choose its standard representative with first nonzero coordinate equal to $1$, and list the resulting representatives as $P_1,\dots,P_n$.

Let $\Fq[x_0,\dots,x_m]_v$ be the $\Fq$-vector space of homogeneous polynomials of degree $v$. If $f\in \Fq[x_0,\dots,x_m]_v$, define
\[
\ev(f)=\bigl(f(P_1),\dots,f(P_n)\bigr).
\]
Then the projective Reed-Muller code is
\[
\PRM(q,m,v)=\{\ev(f):f\in \Fq[x_0,\dots,x_m]_v\}.
\]

First consider the full set of degree-$v$ monomials
\begin{equation}\label{eq:G1-definition}
\mathcal G_1:=\{g=x_0^{a_0}\cdots x_m^{a_m}: a_0+\cdots+a_m=v,\ 0\le a_0,\dots,a_m\}.
\end{equation}
This is the monomial $\Fq$-basis of the homogeneous component. Equivalently,
\[
\Fq[x_0,\dots,x_m]_v=\Span_{\Fq}(\mathcal G_1).
\]
Since $\PRM(q,m,v)=\ev(\Fq[x_0,\dots,x_m]_v)$, the evaluation vectors
$\{\ev(g):g\in\mathcal G_1\}$ span $\PRM(q,m,v)$, but they need not be
linearly independent. Set
\[
\ell:=|\mathcal G_1|=\binom{m+v}{v}.
\]
Order degree-$v$ monomials by descending lexicographic order on exponent vectors, with variable order $x_0\succ x_1\succ\cdots\succ x_m$, and list $\mathcal G_1$ as $g_1,\dots,g_\ell$ in this order.
To obtain a basis after evaluation, we use the following reduced monomial set of
Beelen, Datta, and Ghorpade \cite[Section~3.1]{BeelenDattaGhorpade22}:
\begin{equation}\label{eq:G-definition}
\mathcal G=\Bigl\{\, e=x_0^{a_0}\cdots x_m^{a_m}: a_0+\cdots+a_m=v,\ 0\le a_0,\dots,a_m,\ a_t\le q-1\ \forall\,t>\min\{i:a_i>0\}\Bigr\}.
\end{equation}
As monomials, $\mathcal G$ is a basis of the subspace
$\Span_{\Fq}(\mathcal G)\subseteq \Fq[x_0,\dots,x_m]_v$; its role is that
the evaluation vectors $\{\ev(g):g\in\mathcal G\}$ form the code basis used
below.

The following cardinality formula is S\o rensen's dimension formula applied to $\mathcal G$.

\begin{lemma}[{\cite[Theorem~1]{Sorensen}}]\label{lem:G-cardinality}
We have
\[
k_{q,m}(v):=\dim \PRM(q,m,v)
=
\sum_{\substack{0<t\le v\\ t\equiv v\;(\mathrm{mod}\ q-1)}}
\ \sum_{j=0}^{m}
(-1)^j\binom{m+1}{j}\binom{t-jq+m}{m}.
\]
\end{lemma}

\begin{proposition}[{\cite[Sections~3.1 and~7.1]{BeelenDattaGhorpade22}}]\label{prop:G-basis}
The set
\[
\{\ev(g):g\in \mathcal G\}
\]
forms a basis of $\PRM(q,m,v)$.
\end{proposition}

For use below, list the elements of $\mathcal G$ in the induced lexicographic order:
\[
\mathcal G=\{e_1,\dots,e_k\}.
\]
Then
\[
k=\#\mathcal G=\dim \PRM(q,m,v)=k_{q,m}(v).
\]
The corresponding generator matrix is
\begin{equation}\label{eq:generator-matrix}
G=
\begin{pmatrix}
\ev(e_1)\\
\ev(e_2)\\
\vdots\\
\ev(e_k)
\end{pmatrix}
=
\begin{pmatrix}
e_1(P_1) & e_1(P_2) & \cdots & e_1(P_n)\\
e_2(P_1) & e_2(P_2) & \cdots & e_2(P_n)\\
\vdots & \vdots & \ddots & \vdots\\
e_k(P_1) & e_k(P_2) & \cdots & e_k(P_n)
\end{pmatrix}.
\end{equation}
Let
\[
P_m'=\{P_1,\dots,P_n\}
\]
be the set of all these standard representatives, listed in a fixed order. Different orderings of these standard representatives only permute the coordinates of the evaluation vectors, so they give permutation-equivalent codes and the same Gram matrix.

We also use the full monomial evaluation matrix
\begin{equation}\label{eq:full-monomial-matrix}
G_1=
\begin{pmatrix}
\ev(g_1)\\
\ev(g_2)\\
\vdots\\
\ev(g_\ell)
\end{pmatrix}
=
\begin{pmatrix}
g_1(P_1) & g_1(P_2) & \cdots & g_1(P_n)\\
g_2(P_1) & g_2(P_2) & \cdots & g_2(P_n)\\
\vdots & \vdots & \ddots & \vdots\\
g_\ell(P_1) & g_\ell(P_2) & \cdots & g_\ell(P_n)
\end{pmatrix}.
\end{equation}
The matrices $G$ and $G_1$ span the same row space.

We use the following two results of S\o rensen.

\begin{theorem}[{\cite[Theorem~1]{Sorensen}}]\label{thm:sorensen-dim}
For $1\le v\le m(q-1)$, the projective Reed-Muller code $\PRM(q,m,v)$ has dimension
\[
k_{q,m}(v)
=
\sum_{\substack{0<t\le v\\ t\equiv v\;(\mathrm{mod}\ q-1)}}
\ \sum_{j=0}^{m}
(-1)^j\binom{m+1}{j}\binom{t-jq+m}{m}.
\]
\end{theorem}

Duals are taken with respect to the standard inner product
\[
\langle x,y\rangle=\sum_{i=1}^{n}x_iy_i.
\]

\begin{theorem}[{\cite[Theorem~2]{Sorensen}}]\label{thm:sorensen-dual}
For $1\le v\le m(q-1)$, put
\[
\mu=m(q-1)-v.
\]
Then
\[
\PRM(q,m,v)^\perp=
\begin{cases}
\PRM(q,m,\mu), & v\not\equiv0\pmod{q-1},\\[1ex]
\Span_{\Fq}\{\mathbf 1,\PRM(q,m,\mu)\}, & v\equiv0\pmod{q-1},
\end{cases}
\]
where $\mathbf 1=(1,\dots,1)\in \Fq^n$.
\end{theorem}

If $v\ge m(q-1)+1$, then $\PRM(q,m,v)=\Fq^n$ by \cite[Remark~3]{Sorensen}. Therefore the nontrivial projective range is
\[
1\le v\le m(q-1).
\]

\subsection{The hull of a linear code}

For a linear code $C\subseteq \Fq^n$, its dual code is
\[
C^\perp=\{y\in \Fq^n:\langle x,y\rangle=0\text{ for all }x\in C\}.
\]
The hull of $C$ is defined by
\[
\Hull(C)=C\cap C^\perp.
\]
We use four related situations for projective Reed-Muller codes. The code $C$ is self-dual if $C=C^\perp$, self-orthogonal if $C\subseteq C^\perp$, dual-containing if $C^\perp\subseteq C$, and LCD if $\Hull(C)=0$. For PRM codes, Kaplan and Kim characterized when these cases occur \cite[Theorem~3.2, Theorem~3.5, Corollary~3.8, and Remark~3.11]{KaplanKim}. Song and Luo also pointed out an additional range not covered by these categories: if
\[
\frac{m(q-1)}{2}<v<m(q-1)
\qquad\text{and}\qquad
v\equiv0\pmod{q-1},
\]
then $\PRM(q,m,v)$ is not self-dual, self-orthogonal, LCD, nor dual-containing, but its hull is a projective Reed-Muller code, namely
\[
\Hull\bigl(\PRM(q,m,v)\bigr)=\PRM\bigl(q,m,m(q-1)-v\bigr)
\]
\cite[Theorem~4.1]{SongLuo}.

The following proposition gives the rank relation between the hull and the Gram matrix of a generator matrix.

\begin{proposition}[{\cite[Proposition~3.1]{GuendaJitmanGulliver}}]\label{prop:hull-rank}
Let $C\subseteq \Fq^n$ be a $k$-dimensional linear code, and let $G$ be a generator matrix of $C$. Then
\[
\rank(GG^T)=k-\dim \Hull(C).
\]
In particular, $C$ is self-orthogonal if and only if $GG^T=0$, and it is LCD if and only if $GG^T$ is invertible.
\end{proposition}

\begin{proof}
This is the cited hull-rank identity.
\end{proof}

We use $G_1$ instead of $G$. The following lemma shows that this does not change the Gram rank.

\begin{lemma}[{\cite[Lemma~3.2]{SongLuo}}]\label{lem:same-row-space}
Let $G$ be a full-row-rank matrix, and let $G_1$ be a matrix with the same row space as $G$. Then
\[
\rank(GG^T)=\rank(G_1G_1^T).
\]
\end{lemma}

\begin{proof}
Since the row spaces are equal and $G$ has full row rank, there exists a matrix $P$ of full column rank such that $G_1=PG$. Therefore
\[
G_1G_1^T=P(GG^T)P^T.
\]
It follows that $\rank(G_1G_1^T)\le \rank(GG^T)$. Since $P$ has full column rank, it has a left inverse $L$ with $LP=I$. Therefore
\[
GG^T=L(G_1G_1^T)L^T,
\]
and hence $\rank(GG^T)\le \rank(G_1G_1^T)$. The two ranks are equal.
\end{proof}

Applying Proposition~\ref{prop:hull-rank} and Lemma~\ref{lem:same-row-space} to $G$ and $G_1$, we obtain
\begin{equation}\label{eq:hull-rank-G1}
\rank(G_1G_1^T)
=
\dim \PRM(q,m,v)-\dim \Hull\bigl(\PRM(q,m,v)\bigr).
\end{equation}

For $v\in I_r$ and $m\ge r+1$, put
\[
\Delta_r(v):=\rank(G_1G_1^T).
\]
By \eqref{eq:hull-rank-G1},
\[
\dim \Hull\bigl(\PRM(q,m,v)\bigr)=k_{q,m}(v)-\Delta_r(v).
\]

\subsection{The entry formula}

	For a nonnegative integer $d$, set
	\[
	\sigma(d):=\sum_{\alpha\in\Fq}\alpha^d=
	\begin{cases}
		-1,& d>0\text{ and }\Qm\mid d,\\
		0,& \text{otherwise}.
	\end{cases}
	\]

	The following lemma gives the projective Gram-entry formula used below.

	\begin{lemma}\label{lem:entry-formula}
		If $X^a$ and $X^b$ are monomials of degree $v$ and $c_i:=a_i+b_i$, then, with the convention $0^0=1$,
		\begin{equation}\label{eq:entry-formula}
			\bigl\langle \ev(X^a),\ev(X^b)\bigr\rangle
			=
			\sum_{j=0}^{m}\Bigl(\prod_{i<j}0^{c_i}\Bigr)\prod_{t>j}\sigma(c_t).
		\end{equation}
	\end{lemma}

	\begin{proof}
		Partition the set $P_m'$ of standard representatives according to the position of the first nonzero coordinate. For every $j\in\{0,\dots,m\}$, the representatives whose first nonzero coordinate is the $j$th one are 
		\[
		\{(0,\dots,0,1,\alpha_{j+1},\dots,\alpha_m):\alpha_{j+1},\dots,\alpha_m\in\Fq\}.
		\]
		For such a point, the monomial product $X^aX^b$ takes the value
		\[
		\Bigl(\prod_{i<j}0^{c_i}\Bigr)\,1^{c_j}\,\prod_{t>j}\alpha_t^{c_t},
		\qquad c_i=a_i+b_i.
		\]
		It follows that the sum over this $j$th stratum is
		\begin{align*}
			\sum_{\alpha_{j+1},\dots,\alpha_m\in\Fq}
			\Bigl(\prod_{i<j}0^{c_i}\Bigr)\prod_{t>j}\alpha_t^{c_t}
			&=
			\Bigl(\prod_{i<j}0^{c_i}\Bigr)
			\prod_{t>j}\Bigl(\sum_{\alpha_t\in\Fq}\alpha_t^{c_t}\Bigr)\\
			&=
			\Bigl(\prod_{i<j}0^{c_i}\Bigr)\prod_{t>j}\sigma(c_t).
		\end{align*}
		Summing over all positions $j$ of the first nonzero coordinate gives \eqref{eq:entry-formula}.
	\end{proof}

	The following lemma gives the condition for a nonzero entry in the open-interval situation treated in this paper.

	\begin{lemma}\label{lem:nonzero-criterion}
		Let $v\in I_r$ and put
		\[
		2v=r\Qm+\beta,\qquad 0<\beta<\Qm.
		\]
		Let $X^a$ and $X^b$ be monomials of degree $v$ and set $c_i:=a_i+b_i$. Then
		\[
		\bigl\langle \ev(X^a),\ev(X^b)\bigr\rangle\neq 0
		\]
		if and only if there is a unique index $s$ such that
		\[
		c_0=\cdots=c_{s-1}=0,\qquad c_s>0,\qquad c_s\equiv \beta\pmod{\Qm},
		\]
		\[
		c_j\in \Qm\mathbb Z_{>0}\qquad (j>s).
		\]
		Then the unique surviving term in \eqref{eq:entry-formula} is $j=s$, and
		\[
		\bigl\langle \ev(X^a),\ev(X^b)\bigr\rangle=(-1)^{m-s}.
		\]
	\end{lemma}

	\begin{proof}
		Let $s=\min\{i:c_i>0\}$. If some $j_0>s$ satisfies $c_{j_0}=0$, then every term in \eqref{eq:entry-formula} vanishes: for $j\le s$ the factor $\sigma(c_{j_0})=\sum_{\alpha\in\Fq}\alpha^0=q=0$ occurs, while for $j>s$ the prefix contains $0^{c_s}=0$. Therefore a nonzero entry forces $c_j>0$ for every $j>s$. If some $t>s$ fails to lie in $\Qm\mathbb Z_{>0}$, then $\sigma(c_t)=0$, so every term in \eqref{eq:entry-formula} vanishes: for $j<t$ the factor $\sigma(c_t)$ occurs, while for $j\ge t$ the prefix contains $0^{c_s}$. Therefore a nonzero entry may occur only when
		\[
		c_j\in \Qm\mathbb Z_{>0}\qquad (j>s).
		\]
		Then,
		\[
		c_s=2v-\sum_{j>s}c_j\equiv r\Qm+\beta\equiv \beta\pmod{\Qm}.
		\]
		Since $0<\beta<\Qm$, we obtain $c_s>0$ and $\sigma(c_s)=0$. Therefore in \eqref{eq:entry-formula} every term with $j<s$ vanishes since the product over $t>j$ contains $\sigma(c_s)$, and every term with $j>s$ vanishes since the prefix contains $0^{c_s}$. Therefore the surviving term is $j=s$, and its value is
		\[
		\prod_{j>s}\sigma(c_j)=(-1)^{m-s}.
		\]
		This proves the criterion and the value formula.
	\end{proof}

	\begin{remark}\label{rem:cancellation}
		The open-interval assumption $0<\beta<\Qm$ is used to force
		\[
		c_s\equiv \beta\not\equiv 0\pmod{\Qm},
		\]
	so $\sigma(c_s)=0$ and the surviving term in \eqref{eq:entry-formula} is $j=s$. When $\beta=0$, the argument no longer applies. Those boundary cases are the self-orthogonal and dual-containing cases treated by Kaplan and Kim \cite{KaplanKim}, together with the Song and Luo case \cite{SongLuo} where the hull is another projective Reed--Muller code.
	\end{remark}

	\begin{corollary}\label{cor:support-theorem}
		Let $v\in I_r$, and suppose $m\ge r+1$. If
		\[
		\bigl\langle \ev(X^a),\ev(X^b)\bigr\rangle\neq 0,
		\]
		then $a_i=b_i=0$ for every $i<m-r$. It follows that the nonzero principal block of $G_1G_1^T$ has index set consisting of monomials in the last $r+1$ variables.
	\end{corollary}

	\begin{proof}
		By Lemma~\ref{lem:nonzero-criterion}, if the first positive coordinate occurs at position $t$, then the $m-t$ later coordinates are positive multiples of $\Qm$. Since $2v<(r+1)\Qm$, there can be at most $r$ such later coordinates. Therefore $m-t\le r$, so $t\ge m-r$. Therefore every coordinate before $m-r$ is zero in both monomials.
	\end{proof}

	In what follows, assume $m\ge r+1$ and relabel the last $r+1$ active variables as
	\[
	z_0=x_{m-r},\ z_1=x_{m-r+1},\ \dots,\ z_r=x_m.
	\]
	We define
	\begin{equation}\label{eq:ActSet-definition}
	\ActSet_r(v):=\{M:\exists N,\ \langle \ev(M),\ev(N)\rangle\neq 0\},
	\end{equation}
	where $M$ and $N$ run through all monomials of degree $v$ in $z_0,\dots,z_r$.
	After fixing any order on $\ActSet_r(v)$, we define
	\begin{equation}\label{eq:Sr-definition}
	S_r^{\mathrm{sym}}(v):=\bigl(\langle \ev(M),\ev(N)\rangle\bigr)_{M,N\in\ActSet_r(v)}.
	\end{equation}
	This is the corresponding principal nonzero block.

	\begin{proposition}\label{prop:ambient-independence}
		Let $m,m'\ge r+1$ and $v\in I_r$. After relabeling the active variables in dimension $m$ and in dimension $m'$ as $z_0,\dots,z_r$, the two nonzero principal blocks are identical. It follows that
		\[
		\rank S_r^{\mathrm{sym}}(v)
		\]
		depends only on $(q,r,v)$, so the notation $\Delta_r(v)$ is well defined.
	\end{proposition}

	\begin{proof}
		By Corollary~\ref{cor:support-theorem}, the nonzero block has index set consisting of monomials of degree $v$ in the active variables $z_0,\dots,z_r$. For two such monomials, the corresponding exponent-sum array is the local array $(c_0,\dots,c_r)$. Lemma~\ref{lem:nonzero-criterion} shows that the corresponding entry is determined entirely by this local array. Therefore the entry formula is the same in every ambient dimension once the active variables are relabeled, so the two support blocks agree entry by entry.
	\end{proof}

	It follows that
	\[
	\Delta_r(v)=\rank S_r^{\mathrm{sym}}(v).
	\]

	We use the notation $S_r^{\mathrm{sym}}(v)$ and $\Delta_r(v)$ throughout. The remaining local sets and maps will be introduced when needed.

\subsection{Known base intervals}\label{subsec:prototype-intervals}

The recursion below is anchored at the first two lower intervals, which we record here.

\begin{corollary}\label{cor:prototype-I0}
	For $v\in I_0$,
	\[
	S_0^{\mathrm{sym}}(v)=\begin{bmatrix}1\end{bmatrix},
	\qquad
	\Delta_0(v)=1.
	\]
\end{corollary}

\begin{proof}
	By Corollary~\ref{cor:support-theorem}, only the variable $z_0$ can occur, so the support block is the $1\times1$ matrix $\bigl[\begin{smallmatrix}1\end{smallmatrix}\bigr]$.
\end{proof}

\begin{corollary}\label{cor:prototype-I1}
	For $v\in I_1$,
	\[
	\Delta_1(v)=2v-\Qm+1=A_1(v).
	\]
\end{corollary}

\begin{proof}
	This is the known lower-interval formula in the case $I_1$; see \cite{KaplanKim}.
\end{proof}

\section{The top layer}
	
	We now define the top layer and prove that its Gram block is full rank.
	
	Let $v\in I_r$ with $r\ge 1$, and put
	\[
	2v=r\Qm+\beta,\qquad 0<\beta<\Qm.
	\]
	To describe the top layer, we separate the tail exponents from the first active exponent. Therefore, for a tail vector
	\[
	a=(a_1,\dots,a_r)\in[0,\Qm]^r,
	\]
	denote
	\[
	|a|:=a_1+\cdots+a_r
	\]
	and encode the corresponding monomial of degree $v$ as
	\[
	z_0^{\,v-|a|}z_1^{a_1}\cdots z_r^{a_r}.
	\]
	The bounds
	\[
	L:=r\Qm-v,\qquad U:=v
	\]
	are the right bounds for two simple reasons. First, $|a|\le U$ is the condition that the $z_0$-exponent $v-|a|$ be nonnegative. Second, the complement map
	\[
	a\longmapsto \bar a:=(\Qm-a_1,\dots,\Qm-a_r)
	\]
	satisfies $|\bar a|=r\Qm-|a|$, so the condition $|a|\ge L$ is equivalent to saying that the corresponding complement-partner monomial also has nonnegative $z_0$-exponent. Since $2v=r\Qm+\beta$, we obtain $L=v-\beta$, and hence
	\[
	L\le |a|\le U
	\qquad\Longleftrightarrow\qquad
	0\le v-|a|\le \beta.
	\]
	Thus $[L,U]$ is the range of tail sums preserved by the complement map.

	The standard top layer is defined by
	\begin{equation}\label{eq:TopSet-definition}
	\TopSet_r(v):=\Bigl\{z_0^{v-|a|}z_1^{a_1}\cdots z_r^{a_r}:
	0\le a_i\le \Qm,\ L\le |a|\le U\Bigr\},
	\end{equation}
	and denote
	\[
	A_r(v):=|\TopSet_r(v)|
	=\#\bigl\{a\in[0,\Qm]^r:L\le |a|\le U\bigr\}.
	\]
	
	\begin{proposition}\label{prop:A-explicit}
		For $r\ge1$ and $v\in I_r$, we have
		\[
		A_r(v)=\sum_{t=L}^{U}\sum_{j=0}^{r}(-1)^j\binom{r}{j}\binom{t-jq+r-1}{r-1},
		\]
		where, by convention, $\binom{n}{k}=0$ if $n<k$ or $k<0$.
	\end{proposition}
	
	\begin{proof}
		For a fixed $t$ with $L\le t\le U$, the number of tail vectors $a=(a_1,\dots,a_r)\in[0,\Qm]^r$ with $|a|=t$ is
		\[
		[x^t](1+x+\cdots+x^{\Qm})^r.
		\]
		Using
		\[
		1+x+\cdots+x^{\Qm}=\frac{1-x^q}{1-x},
		\]
		we obtain
		\[
		(1+x+\cdots+x^{\Qm})^r=(1-x^q)^r(1-x)^{-r}
		=\left(\sum_{j=0}^{r}(-1)^j\binom{r}{j}x^{jq}\right)
		\left(\sum_{n\ge 0}\binom{n+r-1}{r-1}x^n\right).
		\]
		It follows that
		\[
		[x^t](1+x+\cdots+x^{\Qm})^r
		=
		\sum_{j=0}^{r}(-1)^j\binom{r}{j}\binom{t-jq+r-1}{r-1}.
		\]
		After summing over $t=L,\dots,U$, we obtain the formula for $A_r(v)$.
	\end{proof}
	
	For $a=(a_1,\dots,a_r)$ in this interval, let
	\[
	M_a:=z_0^{v-|a|}z_1^{a_1}\cdots z_r^{a_r}
	\]
	and write its complement partner as
	\[
	\bar a:=(\Qm-a_1,\dots,\Qm-a_r),\qquad
	M_{\bar a}=z_0^{|a|-L}z_1^{\Qm-a_1}\cdots z_r^{\Qm-a_r}.
	\]
		Since $L\le |a|\le U$ if and only if $L\le r\Qm-|a|\le U$, the tail complement $a\mapsto \bar a$ induces a bijection $M_a\mapsto M_{\bar a}$ of the top layer. Therefore the family $(M_{\bar a})_a$ is a permutation of the top layer.
		
		\begin{definition}\label{def:row-column-layers}
			For every $s\in\{L,\dots,U\}$, the \emph{row layer} $s$ is the family of top-row monomials $M_a$ with $|a|=s$, while the \emph{column layer} $s$ is the family of top-column complement partners $M_{\bar b}$ with $|b|=s$. Equivalently, row layer $s$ consists of top-row monomials with $z_0$-exponent $v-s$, and column layer $s$ consists of top-column monomials with $z_0$-exponent $s-L$.
			
			In particular, $|a|=U$ means that the row index lies in the \emph{upper row boundary}, i.e.\ among the top-row monomials $M_a$ with $z_0$-exponent $0$. The condition $|b|=L$ means that the column index lies in the \emph{lower column boundary}, i.e.\ among the top-column monomials $M_{\bar b}$ with $z_0$-exponent $0$. Finally, $|b|=U$ means that the column index lies in the \emph{upper column layer}, i.e.\ among the top-column monomials $M_{\bar b}$ with $z_0$-exponent $U-L=\beta$.
		\end{definition}
		
		\begin{theorem}\label{thm:top-full-rank}
			For every $s$ with $L\le s\le U$, let
		\[
		n_s:=\#\{a\in[0,\Qm]^r:|a|=s\}.
		\]
		We index the rows of $P_r(v)$ by all tuples
		\[
		a=(a_1,\dots,a_r)\in[0,\Qm]^r,\qquad L\le |a|\le U,
		\]
		ordered first by increasing $|a|$ and then lexicographically inside each level set $\{|a|=s\}$. We index the columns by all tuples
		\[
		b=(b_1,\dots,b_r)\in[0,\Qm]^r,\qquad L\le |b|\le U,
		\]
		in the same way, and attach to the column indexed by $b$ the complement partner monomial
		\[
		M_{\bar b}=z_0^{|b|-L}z_1^{\Qm-b_1}\cdots z_r^{\Qm-b_r}.
		\]
		It follows that
		\begin{equation}\label{eq:Pr-definition}
		P_r(v):=\bigl(\langle \ev(M_a),\ev(M_{\bar b})\rangle\bigr)_{a,b}
		\end{equation}
		is an $A_r(v)\times A_r(v)$ matrix, with block decomposition
		\[
		P_r(v)=\bigl(P_{s,t}\bigr)_{L\le s,t\le U},\qquad
		P_{s,t}\in M_{n_s\times n_t}(\Fq).
		\]
		Then $P_r(v)$ is a block lower triangular matrix, every diagonal block is
		\[
		P_{s,s}=(-1)^rI_{n_s},
		\]
			and the possible off-diagonal block is the corner block $P_{U,L}$, from the upper row boundary $|a|=U$ to the lower column boundary $|b|=L$. In particular,
		\[
		\rank P_r(v)=A_r(v).
		\]
		It follows that the principal submatrix of $S_r^{\mathrm{sym}}(v)$ indexed by $\TopSet_r(v)$ has rank $A_r(v)$.
	\end{theorem}
	
	\begin{proof}
		For a row index $a$ and a column index $b$, let $c_i$ be the exponent sums for $M_a$ and $M_{\bar b}$. Then
		\[
		c_0=(v-|a|)+(|b|-L)=\beta+|b|-|a|,
		\]
		and for $1\le i\le r$,
		\[
		c_i=a_i+(\Qm-b_i)=\Qm+a_i-b_i.
		\]
		For $a=b$,
		\[
		(c_0,c_1,\dots,c_r)=(\beta,\Qm,\dots,\Qm),
		\]
		and Lemma~\ref{lem:nonzero-criterion} gives
		\[
		\langle \ev(M_a),\ev(M_{\bar a})\rangle=(-1)^r.
		\]
		It follows that every diagonal block $P_{s,s}$ is the scalar matrix $(-1)^rI_{n_s}$.
		
		Now assume $a\neq b$ and that the entry is nonzero. If $c_0>0$, then Lemma~\ref{lem:nonzero-criterion} forces every $c_i$ with $i\ge 1$ to be a positive multiple of $\Qm$. Since $0\le c_i\le 2\Qm$, each $c_i$ is either $\Qm$ or $2\Qm$, so every $a_i-b_i$ is either $0$ or $\Qm$. Therefore $|a|-|b|$ is a nonnegative multiple of $\Qm$. By contrast,
		\[
		0<c_0=\beta+|b|-|a|<\Qm
		\]
		since $0<\beta<\Qm$. Therefore $0\le |a|-|b|<\Qm$, and the only nonnegative multiple of $\Qm$ in this interval is $0$. Therefore $|a|=|b|$ and all $a_i=b_i$, contradicting $a\neq b$.
		
		Thus every nonzero off-diagonal entry must satisfy $c_0=0$, i.e.,
		\[
		|a|=|b|+\beta.
		\]
			Since $U-L=\beta$, this can happen only for the corner pair $(|a|,|b|)=(U,L)$, namely from the upper row boundary to the lower column boundary. Therefore
		\[
		P_{s,t}=0\qquad\text{whenever }s<t\text{ or }(s,t)\neq(U,L)\text{ with }s\neq t.
		\]
		It follows that $P_r(v)$ is block lower triangular with diagonal blocks $(-1)^rI_{n_s}$, and so
		\[
		\rank P_r(v)=\sum_{s=L}^{U}n_s=A_r(v).
		\]
		Since the column family $(M_{\bar b})_b$ is a permutation of the top-layer monomials, the corresponding principal block of $S_r^{\mathrm{sym}}(v)$ has the same rank.
	\end{proof}

	\section{The rightmost lift}
	
	We define the rightmost lift and prove the property needed below.
	
	\begin{definition}
		Let
		\[
		M=z_0^{a_0}z_1^{a_1}\cdots z_s^{a_s}
		\]
		be a nonconstant monomial. Define
		\[
		\rho(M):=\max\{j:a_j>0\},
		\]
		and let
		\[
		\Lift(M):=z_{\rho(M)}^{\Qm}M.
		\]
	\end{definition}
	
	\begin{proposition}\label{prop:rightmost-lift}
		Let $M$ be any nonconstant monomial. Then $\ev(\Lift(M))=\ev(M)$ on the projective point set under consideration. The map $M\mapsto \Lift(M)$ is injective. It follows that if $\mathcal F$ is any family of nonconstant monomials of degree $u$, then the degree-$(u+\Qm)$ Gram matrix of the lifted family $\Lift(\mathcal F)$ is the same as the degree-$u$ Gram matrix of $\mathcal F$. In particular, the latter occurs as a principal submatrix of the former.
	\end{proposition}
	
	\begin{proof}
		Let $j=\rho(M)$. For every standard projective representative $P$, if the $j$th coordinate is zero then both $M(P)$ and $\Lift(M)(P)$ are zero since $a_j>0$. If the $j$th coordinate is nonzero, then it lies in $\Fq^*$ and so its $\Qm$th power is $1$. Therefore $\Lift(M)(P)=M(P)$ for every $P$.
		
		Injectivity follows from the definition: the rightmost positive variable of $\Lift(M)$ is still $z_j$, and its exponent there is larger by $\Qm$, so $M$ can be recovered by subtracting $\Qm$ from that exponent. The Gram-matrix claim follows since every lifted row vector and every lifted column vector equals the original one.
	\end{proof}
	
	Let $r\ge 2$ and $v\in I_r$, and set
	\[
	u:=v-\Qm\in I_{r-2}.
	\]
	By Corollary~\ref{cor:support-theorem}, the lower support block $S_{r-2}^{\mathrm{sym}}(u)$ has index set consisting of monomials in the variables $z_2,\dots,z_r$. Applying Proposition~\ref{prop:rightmost-lift} to that lower family and padding by the two leading zeros $z_0^0z_1^0$, we obtain the following exact embedding.
	
	\begin{corollary}\label{cor:lower-embedding}
		The block $S_{r-2}^{\mathrm{sym}}(u)$ occurs as a principal submatrix of $S_r^{\mathrm{sym}}(v)$. More precisely, if $L$ runs through the lower index monomials in $z_2,\dots,z_r$, then the lifted monomials $\Lift(L)$ are pairwise distinct, lie in $\ActSet_r(v)$, and satisfy
		\[
		S_r^{\mathrm{sym}}(v)\big|_{\Lift(\ActSet_{r-2}(u))\times \Lift(\ActSet_{r-2}(u))}=S_{r-2}^{\mathrm{sym}}(u).
		\]
	\end{corollary}
	
	\begin{proof}
		Let $L\in \ActSet_{r-2}(u)$. From the definition of the lower index set, there exists a lower monomial $L'$ with
		\[
		\langle \ev(L),\ev(L')\rangle\neq 0.
		\]
		By Proposition~\ref{prop:rightmost-lift},
		\[
		\langle \ev(\Lift(L)),\ev(\Lift(L'))\rangle
		=\langle \ev(L),\ev(L')\rangle\neq 0,
		\]
		so every lifted monomial belongs to $\ActSet_r(v)$. The proposition implies that $L\mapsto \Lift(L)$ is injective, so the lifted index set is pairwise distinct.
		
		Order the lifted rows and columns by the original lower order. For every lower index monomials $L_1,L_2$ we obtain
		\[
		\bigl\langle \ev(\Lift(L_1)),\ev(\Lift(L_2))\bigr\rangle
		=\langle \ev(L_1),\ev(L_2)\rangle.
		\]
		It follows that the restricted principal submatrix of $S_r^{\mathrm{sym}}(v)$ on the lifted index set equals $S_{r-2}^{\mathrm{sym}}(u)$ entry by entry. In particular, the embedding is exact, not only equivalent up to permutation or rank.
	\end{proof}

	Recall \eqref{eq:ActSet-definition} and \eqref{eq:TopSet-definition} that $\ActSet_r(v)$ indexes the whole nonzero principal block $S_r^{\mathrm{sym}}(v)$, while $\TopSet_r(v)\subseteq \ActSet_r(v)$ is the top layer. Denote
	\begin{equation}\label{eq:RemSet-definition}
	\RemSet_r(v):=\ActSet_r(v)\setminus \TopSet_r(v).
	\end{equation}
	for the remainder index set.
	
\begin{corollary}\label{cor:lift-remainder}
		Every monomial in $\Lift(\ActSet_{r-2}(u))$ has $z_0$- and $z_1$-exponents equal to $0$ and has a tail exponent greater than $\Qm$. It follows that
		\[
		\Lift(\ActSet_{r-2}(u))\subseteq \RemSet_r(v)
		\qquad\text{and}\qquad
		\Lift(\ActSet_{r-2}(u))\cap \TopSet_r(v)=\varnothing.
		\]
	\end{corollary}
	
	\begin{proof}
		A monomial in the lower index set uses only the variables $z_2,\dots,z_r$. Its rightmost lift hence keeps the $z_0$- and $z_1$-exponents equal to $0$ and adds $\Qm$ to the exponent of a variable that was already present. Therefore one tail exponent becomes strictly larger than $\Qm$. Such a monomial cannot lie in the top layer, since every top-layer monomial has all tail exponents at most $\Qm$.
	\end{proof}

	\section{The global interval recursion}
	
	Fix $r\ge 2$ and $v\in I_r$, and put $u=v-\Qm\in I_{r-2}$.
	We start with the symmetric support block $S_r^{\mathrm{sym}}(v)$ in the local variables $z_0,\dots,z_r$, but the Schur-complement argument is carried out after passing to a one-sided column-permuted working block.
	Thus the top rows remain indexed by the family $(M_a)$, whereas the top columns will later be reindexed by the pairing $M_a\mapsto M_{\bar a}$.
	
	\subsection{The remainder block}
	
	Recall that $\RemSet_r(v)$ denotes the remainder index set defined in \eqref{eq:RemSet-definition}.
	The following lemma describes these remainder monomials.
	
	\begin{lemma}\label{lem:remainder-shape}
		If
		\[
		M=z_0^{h}z_1^{a_1}\cdots z_r^{a_r}\in \ActSet_r(v)
		\]
		and $h>0$, then $M\in \TopSet_r(v)$. It follows that every monomial in $\RemSet_r(v)$ has $z_0$-exponent $0$ and at least one tail exponent greater than $\Qm$.
	\end{lemma}
	
	\begin{proof}
		Choose a monomial $N=z_0^{b_0}\cdots z_r^{b_r}$ such that
		\[
		\langle \ev(M),\ev(N)\rangle\neq 0.
		\]
		Set $c_i=a_i+b_i$. Since $h>0$, we obtain $c_0>0$. Lemma~\ref{lem:nonzero-criterion} then forces each tail coordinate $c_i$ with $1\le i\le r$ to be a positive multiple of $\Qm$, hence at least $\Qm$. Therefore
		\[
		2v-c_0=\sum_{i=1}^r c_i\ge r\Qm.
		\]
		Since $2v=r\Qm+\beta$, we obtain $c_0\le \beta$. Also $c_0\equiv 2v\equiv \beta\pmod{\Qm}$, and $0<c_0<\Qm$, so in fact $c_0=\beta$. Therefore
		\[
		\sum_{i=1}^r c_i=r\Qm.
		\]
		Since each tail coordinate is at least $\Qm$, every one of them must be $\Qm$. It follows that $a_i\le \Qm$ for all $i\ge 1$, while $0<h\le c_0=\beta$. Therefore
		\[
		L=r\Qm-v\le v-h\le v,
		\]
		hence $M\in \TopSet_r(v)$.
		
		The consequence follows directly: if $M\in \RemSet_r(v)$ then $h=0$. Since $M\notin \TopSet_r(v)$, not all tail exponents can be at most $\Qm$, so at least one of them is greater than $\Qm$.
	\end{proof}
	
	We now make this passage to the working block precise. Recall that
\[
\ActSet_r(v)=\TopSet_r(v)\sqcup \RemSet_r(v).
\]
We use the following order on $\ActSet_r(v)$ the \emph{symmetric support order}: list first the monomials in $\TopSet_r(v)$, then the monomials in $\RemSet_r(v)$, and within each block use lexicographic order. If we index both the rows and the columns of $S_r^{\mathrm{sym}}(v)$ by this symmetric support order, then the symmetric support block splits as
	\[
	S_r^{\mathrm{sym}}(v)=
	\begin{pmatrix}
		T & B_0\\
		B_0^{\mathsf T} & D
	\end{pmatrix},
	\]
	where $T$ denotes the principal block on the top layer, $D$ is the principal block on the remainder, and $B_0$ contains the mixed top/remainder entries. Let $\Pi$ be the permutation matrix on the top-layer index set induced by the bijection $M_a\mapsto M_{\bar a}$ on the top layer (equivalently, by the tail complement $a\mapsto \bar a$). We keep the row order fixed and apply this permutation only to the top columns. Then
	\[
	W_r(v)=
	\begin{pmatrix}
		T\Pi & B_0\\
		B_0^{\mathsf T}\Pi & D
	\end{pmatrix},
	\]
	and we hence write
	\begin{equation}\label{eq:block-decomposition}
		W_r(v):=
		S_r^{\mathrm{sym}}(v)
		\begin{pmatrix}
			\Pi & 0\\
			0 & I
		\end{pmatrix}
		=
		\begin{pmatrix}
			P & B\\
			C & D
		\end{pmatrix},
	\end{equation}
	where
	\[
	P=P_r(v)=T\Pi,\qquad B=B_0,\qquad C=B_0^{\mathsf T}\Pi.
	\]
It follows that for every $M_b\in\TopSet_r(v)$ and every $Y\in\RemSet_r(v)$, we have $B_{M_{\bar b},\,Y}=C_{Y,\,M_{\bar b}}$.
	Since the right factor is invertible,
	\[
	\rank W_r(v)=\rank S_r^{\mathrm{sym}}(v).
	\]
		The decomposition \eqref{eq:block-decomposition} places the complement-paired top-layer matrix $P_r(v)$ in the upper-left block and leaves the lower-right block unchanged. Since the permutation acts on one side only, the working block $W_r(v)$ need not be symmetric, even though the original support block $S_r^{\mathrm{sym}}(v)$ is symmetric as a principal submatrix of the Gram matrix $G_1G_1^T$. Therefore $C$ is obtained from $B^{\mathsf T}$ by reindexing the top-layer columns via the pairing $M_a \leftrightarrow M_{\bar a}$.
		
		Below, top rows are always read in the row-layer notation of Definition~\ref{def:row-column-layers}, so they have index sets consisting of monomials $M_a$. In contrast, after the one-sided permutation the top columns are read in the column-layer notation, so they are indexed by complement partners $M_{\bar b}$. In particular, statements about the support of $P$, $B$, $C$, or products such as $P^{-1}B$ must always specify whether the support is on the row side or on the column side.

	\subsection{Boundary support of the off-diagonal blocks}
	
	\begin{lemma}\label{lem:B-support}
		Let a row of $B$ have index a top-row monomial $M_a\in\TopSet_r(v)$ and a column by a remainder monomial $Y\in\RemSet_r(v)$. If $B_{M_a,Y}\neq 0$, then $|a|=U=v$. Equivalently, every nonzero row of $B$ lies in the upper row boundary, i.e.\ is indexed by a top-row monomial $M_a$ with $z_0$-exponent $0$.
	\end{lemma}
	
	\begin{proof}
		Let $h=v-|a|$ be the $z_0$-exponent of $M_a$. If $B_{M_a,Y}\neq 0$ and $h>0$, then the first exponent sum is $c_0=h$. As in Lemma~\ref{lem:remainder-shape}, we obtain $h=\beta$ and all tail exponent sums equal $\Qm$. Therefore the column monomial $Y$ must be the complement partner $M_{\bar a}$ of $M_a$. But $M_{\bar a}$ belongs to the permuted top family, not to the remainder set $\RemSet_r(v)$. This contradicts the choice of the column index $Y$. Therefore $h=0$, so $|a|=v=U$.
	\end{proof}
	
	\begin{lemma}\label{lem:C-support}
		Let a row of $C$ have index a remainder monomial $Y\in\RemSet_r(v)$ and a column by a top-column complement partner $M_{\bar b}$ with $M_b\in\TopSet_r(v)$. If $C_{Y,M_{\bar b}}\neq 0$, then $|b|=L=r\Qm-v$. Equivalently, every nonzero column of $C$ lies in the lower column boundary, i.e.\ is indexed by a complement partner $M_{\bar b}$ with $z_0$-exponent $0$.
	\end{lemma}
	
	\begin{proof}
		If $C_{Y,M_{\bar b}}\neq 0$, then the relation $B_{M_{\bar b},\,Y}=C_{Y,\,M_{\bar b}}$ after \eqref{eq:block-decomposition} gives $B_{M_{\bar b},\,Y}\neq 0$. Therefore Lemma~\ref{lem:B-support} implies $|\bar b|=U=v$. Since $|\bar b|=r\Qm-|b|$, we obtain $|b|=r\Qm-U=r\Qm-v=L$.
	\end{proof}

	\begin{lemma}\label{lem:schur-zero}
		For the block decomposition \eqref{eq:block-decomposition},
		\[
		CP^{-1}B=0.
		\]
	\end{lemma}
	
	\begin{proof}
		By Theorem~\ref{thm:top-full-rank}, after grouping the top rows by row layers and the top columns by column layers, $P$ is block lower triangular with diagonal blocks $(-1)^rI$ and a single possible off-diagonal corner block from the upper row boundary $|a|=U$ to the lower column boundary $|b|=L$. Therefore we write $P=D_0+E$, where $D_0$ is block diagonal and $E$ has support only in that corner block. Since $E^2=0$, we have $P^{-1}=D_0^{-1}-D_0^{-1}ED_0^{-1}$. By Lemma~\ref{lem:B-support}, the nonzero rows of $B$ lie in the upper row boundary. Since $D_0^{-1}$ is block diagonal for the layer decomposition $L,L+1,\dots,U$, left multiplication by $D_0^{-1}$ transfers this support to the corresponding upper column layer: the rows of $D_0^{-1}B$ are indexed on the column side, and only the layer $|b|=U$ may remain nonzero. As $E$ has nonzero columns only in the lower column boundary $|b|=L$, we obtain $ED_0^{-1}B=0$. Therefore $P^{-1}B=D_0^{-1}B$, so the nonzero rows of $P^{-1}B$ also lie in the upper column layer $|b|=U$. By Lemma~\ref{lem:C-support}, the nonzero columns of $C$ lie in the disjoint lower column boundary $|b|=L$. Therefore $CP^{-1}B=0$.
	\end{proof}
	
	\begin{proposition}\label{prop:block-elimination}
		For the working block matrix in \eqref{eq:block-decomposition}, define
		\begin{equation}\label{eq:LR-definition}
		L_r(v):=\begin{pmatrix}I&0\\-CP^{-1}&I\end{pmatrix},
		\qquad
		R_r(v):=\begin{pmatrix}I&-P^{-1}B\\0&I\end{pmatrix}.
		\end{equation}
		Then $L_r(v)$ and $R_r(v)$ are invertible, and
		\[
		\begin{aligned}
			W_r(v)R_r(v)&=\begin{pmatrix}P&0\\C&D\end{pmatrix},\\
			L_r(v)W_r(v)R_r(v)&=\begin{pmatrix}P&0\\0&D\end{pmatrix}.
		\end{aligned}
		\]
		It follows that
		\[
		\rank W_r(v)=\rank P+\rank D.
		\]
		Since $W_r(v)$ differs from $S_r^{\mathrm{sym}}(v)$ by a column permutation and $\rank P=A_r(v)$,
		\[
		\rank S_r^{\mathrm{sym}}(v)=A_r(v)+\rank D.
		\]
	\end{proposition}
	
	\begin{proof}
		A block-matrix computation, using Lemma~\ref{lem:schur-zero}, gives the displayed identities; the rank formulas follow.
	\end{proof}

		\begin{remark}\label{rem:deleted-monomials}
			Since $W_r(v)$ is equivalent to $\operatorname{diag}(P,D)$ and $P$ is invertible, every monomial in $\TopSet_r(v)$ must be kept if we want full rank. It remains only to choose which rows and columns of $D$ to keep. Therefore the recursive choice of rows and columns to delete is the lower-interval selection problem inside $D$.
		\end{remark}
	
	\subsection{The rank of the remainder block}
	
	For $Y\in\RemSet_r(v)$, Lemma~\ref{lem:remainder-shape} gives $z_0$-exponent $0$. Since $Y\in\ActSet_r(v)$ is a monomial of degree $v$ in the active variables $z_0,\dots,z_r$, it has the form
	\[
	Y=z_1^{a_1}\cdots z_r^{a_r},
	\qquad a_1+\cdots+a_r=v.
	\]
	The same lemma shows that some tail exponent is greater than $\Qm$. Let
	\[
	\eta(Y):=\max\{j\in\{1,\dots,r\}:a_j>\Qm\},
	\]
	and define the one-step reduction
	\[
	\Red(Y):=Y/z_{\eta(Y)}^{\Qm}.
	\]
	Then $\Red(Y)$ is a monomial of degree $u$ in the variables $z_1,\dots,z_r$.
	
	\begin{lemma}\label{lem:red-eval}
		For $Y\in\RemSet_r(v)$,
		\[
		\ev(Y)=\ev(\Red(Y)).
		\]
	\end{lemma}
	
	\begin{proof}
		Since $Y=z_{\eta(Y)}^{\Qm}\Red(Y)$ and $a_{\eta(Y)}>\Qm$ implies $z_{\eta(Y)}\mid \Red(Y)$, at each projective point either both monomials vanish or $z_{\eta(Y)}^{\Qm}=1$, so $\ev(Y)=\ev(\Red(Y))$.
	\end{proof}
	
		Recall that $u=v-\Qm$. Set
		\begin{equation}\label{eq:Mru-definition}
		\mathcal M_r(u):=\{z_1^{a_1}\cdots z_r^{a_r}: a_1+\cdots+a_r=u,\ 0\le a_1,\dots,a_r\}.
		\end{equation}
		Fix an order on this set. We then define
		\begin{equation}\label{eq:Hr-definition}
		H_r(u):=\bigl(\langle \ev(M),\ev(M')\rangle\bigr)_{M,M'\in\mathcal M_r(u)}.
		\end{equation}
		It follows that $H_r(u)$ is the full degree-$u$ Gram matrix in the variables $z_1,\dots,z_r$, before any support pruning. It is not one of the blocks in \eqref{eq:block-decomposition}; it is an auxiliary matrix that controls the lower-right block $D$ of $W_r(v)$. More precisely, after introducing the reduction matrix $R$ below, we prove
		\[
		D=R\,H_r(u)\,R^T.
		\]
		
		\begin{lemma}\label{lem:Hr-support}
			The nonzero principal block of $H_r(u)$ is indexed by the monomials of degree $u$ in $z_2,\dots,z_r$. Identifying these monomials with $\ActSet_{r-2}(u)$, we have
			\[
			H_r(u)\big|_{\ActSet_{r-2}(u)\times \ActSet_{r-2}(u)}=S_{r-2}^{\mathrm{sym}}(u).
			\]
			It follows that
			\begin{equation}\label{eq:Hr-rank}
				\rank H_r(u)=\Delta_{r-2}(u).
			\end{equation}
		\end{lemma}
		
		\begin{proof}
			Since $u\in I_{r-2}$, Corollary~\ref{cor:support-theorem} applied to the family of all monomials of degree $u$ in the variables $z_1,\dots,z_r$ implies that a nonzero entry of $H_r(u)$ may involve only monomials in $z_2,\dots,z_r$. Identifying these monomials with the lower index set $\ActSet_{r-2}(u)$, the surviving block is $S_{r-2}^{\mathrm{sym}}(u)$. Therefore the displayed principal-submatrix identity holds, and \eqref{eq:Hr-rank} follows from the definition of $\Delta_{r-2}(u)$.
		\end{proof}
		
		Next set the $0$-$1$ reduction matrix
	\begin{equation}\label{eq:R-definition}
	\begin{aligned}
		R&=(R_{Y,M})_{Y\in\RemSet_r(v),\,M\in\mathcal M_r(u)},\\
		R_{Y,M}&=
		\begin{cases}
			1,& M=\Red(Y),\\
			0,& \text{otherwise}.
		\end{cases}
	\end{aligned}
	\end{equation}
	Each row of $R$ hence has one nonzero entry, in the column indexed by the reduction of the corresponding remainder monomial. If two different remainder monomials have the same reduction, then the corresponding rows of $R$ are equal.
	
	\begin{lemma}\label{lem:repeat-rank}
		Let $M$ be a matrix over a field, and let $R,C$ be matrices of compatible sizes. Then
		\[
		\rank(RMC)\le \rank(M).
		\]
		In particular, repeating rows or columns of a matrix cannot increase its rank.
	\end{lemma}
	
	\begin{proof}
		This follows from the inequality $\rank(RMC)\le \min\{\rank(R),\rank(M),\rank(C)\}$.
	\end{proof}
	
	\begin{lemma}\label{lem:D-upper}
		There is a $0$-$1$ selection matrix $R$ such that
		\[
		D=R\,H_r(u)\,R^T.
		\]
		It follows that
		\[
		\rank D\le \rank H_r(u)=\Delta_{r-2}(u).
		\]
	\end{lemma}
	
	\begin{proof}
		For remainder indices $Y,Y'$, each row of $R$ has a single $1$, so
		\begin{align*}
			(RH_r(u)R^T)_{Y,Y'}
			&=\sum_{M,M'\in\mathcal M_r(u)} R_{Y,M}\,H_r(u)_{M,M'}\,R_{Y',M'}\\
			&=H_r(u)_{\Red(Y),\Red(Y')}\\
			&=\bigl\langle \ev(\Red(Y)),\ev(\Red(Y'))\bigr\rangle.
		\end{align*}
		By Lemma~\ref{lem:red-eval}, this equals
		\[
		\langle \ev(Y),\ev(Y')\rangle=D_{Y,Y'}.
		\]
			It follows that $D=R H_r(u)R^T$. If two different remainder monomials have the same reduction, the corresponding rows of $R$ are equal, and $R$ repeats the associated row and column of $H_r(u)$. Lemma~\ref{lem:repeat-rank} and Lemma~\ref{lem:Hr-support} hence give
			\[
			\rank D\le \rank H_r(u)=\Delta_{r-2}(u).
			\]
	\end{proof}

	\begin{lemma}\label{lem:D-lower}
		The principal submatrix of $D$ indexed by $\Lift(\ActSet_{r-2}(u))$ is $S_{r-2}^{\mathrm{sym}}(u)$. It follows that
		\[
		\rank D\ge \Delta_{r-2}(u).
		\]
	\end{lemma}
	
	\begin{proof}
		By Corollary~\ref{cor:lift-remainder}, the lifted lower index set lies inside $\RemSet_r(v)$, so it indexes a principal submatrix of $D$. Moreover, for every lower-index monomial $L$, we have
		\[
			\Red(\Lift(L))=L.
		\]
			It follows that on rows and columns indexed by $\Lift(\ActSet_{r-2}(u))$, the reduction matrix $R$ restricts to the identity on the corresponding lower-index columns of $H_r(u)$. Therefore
			\[
				D\big|_{\Lift(\ActSet_{r-2}(u))\times \Lift(\ActSet_{r-2}(u))}
				=
				H_r(u)\big|_{\ActSet_{r-2}(u)\times \ActSet_{r-2}(u)}.
			\]
			By Lemma~\ref{lem:Hr-support},
			\[
				H_r(u)\big|_{\ActSet_{r-2}(u)\times \ActSet_{r-2}(u)}=S_{r-2}^{\mathrm{sym}}(u).
			\]
			It follows that
			\[
				D\big|_{\Lift(\ActSet_{r-2}(u))\times \Lift(\ActSet_{r-2}(u))}=S_{r-2}^{\mathrm{sym}}(u),
		\]
		and hence
		\[
			\rank D\ge \rank S_{r-2}^{\mathrm{sym}}(u)=\Delta_{r-2}(u).
		\]
	\end{proof}
	
	\begin{proposition}\label{prop:D-rank}
		In the block decomposition \eqref{eq:block-decomposition},
		\[
		\rank D=\Delta_{r-2}(u).
		\]
	\end{proposition}
	
	\begin{proof}
		Lemma~\ref{lem:D-upper} gives the upper bound, and Lemma~\ref{lem:D-lower} gives the lower bound, so equality follows.
	\end{proof}
	
	We can now prove the global recursion.
	
	\begin{theorem}\label{thm:interval-recursion}
		Let $r\ge 2$ and $v\in I_r$. Then
		\[
		\Delta_r(v)=A_r(v)+\Delta_{r-2}(v-\Qm).
		\]
	\end{theorem}
	
	\begin{proof}
		Put $u=v-\Qm$. By Proposition~\ref{prop:block-elimination},
		\[
		\Delta_r(v)=\rank S_r^{\mathrm{sym}}(v)=A_r(v)+\rank D.
		\]
		By Proposition~\ref{prop:D-rank}, $\rank D=\Delta_{r-2}(u)$, so the recursion follows.
	\end{proof}
	
	\subsection{Recursive full-rank principal blocks}
	
	The previous theorem also gives the recursive principal-block statement needed later for full-rank principal blocks.
	
	\begin{corollary}\label{cor:principal-block-recursion}
		Let $E_{r-2}(u)\subseteq \ActSet_{r-2}(u)$ index a principal submatrix of $S_{r-2}^{\mathrm{sym}}(u)$ of rank $\Delta_{r-2}(u)$. Then the set
		\begin{equation}\label{eq:principal-block-Er-definition}
		E_r(v):=\TopSet_r(v)\cup \Lift(E_{r-2}(u))
		\end{equation}
		indexes a principal submatrix of $S_r^{\mathrm{sym}}(v)$ of rank $\Delta_r(v)$.
	\end{corollary}
	
	\begin{proof}
		By Corollary~\ref{cor:lift-remainder}, the lifted set $\Lift(E_{r-2}(u))$ lies in $\RemSet_r(v)$. Consider the principal submatrix
		\[
			S_r^{\mathrm{sym}}(v)\big|_{(\TopSet_r(v)\cup \Lift(E_{r-2}(u)))\times (\TopSet_r(v)\cup \Lift(E_{r-2}(u)))},
		\]
		and apply to its top columns the complement permutation used in \eqref{eq:block-decomposition}. Since this is only a column permutation, it does not change the rank. The resulting working block has the form
		\begin{equation}\label{eq:restricted-block-decomposition}
			W_{r,E}(v)=
			\begin{pmatrix}
				P & B_E\\
				C_E & D_E
			\end{pmatrix}.
		\end{equation}
		
		The upper-left block remains $P$. Also, since
		\[
			\Lift(E_{r-2}(u))\subseteq \Lift(\ActSet_{r-2}(u)),
		\]
		By Corollary~\ref{cor:lower-embedding},
		\[
			D_E
			=
			S_r^{\mathrm{sym}}(v)\big|_{\Lift(E_{r-2}(u))\times \Lift(E_{r-2}(u))}
			=
			S_{r-2}^{\mathrm{sym}}(u)\big|_{E_{r-2}(u)\times E_{r-2}(u)}.
		\]
		It follows that
		\[
			\rank D_E=\Delta_{r-2}(u).
		\]
		
		It remains to show that the restricted Schur complement still vanishes. Since $B_E$ and $C_E$ are obtained from $B$ and $C$ by restriction to lifted remainder indices, Lemmas~\ref{lem:B-support} and \ref{lem:C-support} place the nonzero rows of $B_E$ in the upper row boundary $|a|=U$ and the nonzero columns of $C_E$ in the lower column boundary $|b|=L$. Write, as in Lemma~\ref{lem:schur-zero},
		\[
			P=D_0+E,
			\qquad
			P^{-1}=D_0^{-1}-D_0^{-1}ED_0^{-1},
		\]
		where $D_0$ is block diagonal and $E$ has support only in the corner block from the upper row boundary to the lower column boundary. Then the rows of $D_0^{-1}B_E$ are read on the column side, and only the upper column layer $|b|=U$ may remain nonzero. Therefore
		\[
			ED_0^{-1}B_E=0.
		\]
		It follows that
		\[
			P^{-1}B_E=D_0^{-1}B_E,
		\]
		and the nonzero rows of $P^{-1}B_E$ lie in the upper column layer $|b|=U$. Since the nonzero columns of $C_E$ lie in the disjoint lower column boundary $|b|=L$, we obtain
		\[
			C_E P^{-1} B_E=0.
		\]
		
		The block elimination in Proposition~\ref{prop:block-elimination} now gives
		\[
			\rank W_{r,E}(v)=\rank P+\rank D_E
			=
			A_r(v)+\Delta_{r-2}(u)
			=
			\Delta_r(v),
		\]
		by Theorem~\ref{thm:interval-recursion}. Since $W_{r,E}(v)$ differs from the restricted principal submatrix of $S_r^{\mathrm{sym}}(v)$ only by a column permutation, the latter also has rank $\Delta_r(v)$.
	\end{proof}
	
	\begin{corollary}\label{cor:recursive-principal-block}
		Define the recursive principal-block family by
		\begin{equation}\label{eq:Er-definition}
		\begin{aligned}
			E_0(v)&:=\{z_0^v\}, && v\in I_0,\\
			E_1(v)&:=\TopSet_1(v), && v\in I_1,\\
			E_r(v)&:=\TopSet_r(v)\cup \Lift\bigl(E_{r-2}(v-\Qm)\bigr), && r\ge 2,\ v\in I_r.
		\end{aligned}
		\end{equation}
		Then $E_r(v)\subseteq \ActSet_r(v)$ indexes a principal submatrix of $S_r^{\mathrm{sym}}(v)$ of rank $\Delta_r(v)$. In particular,
		\[
		|E_r(v)|=\Delta_r(v).
		\]
	\end{corollary}
	
	\begin{proof}
		The statement for $r=0$ follows directly, and for $r=1$ it is Theorem~\ref{thm:top-full-rank} since $\Delta_1(v)=A_1(v)$. The general case is obtained by induction from Corollary~\ref{cor:principal-block-recursion}.
	\end{proof}
	
	\begin{remark}\label{rem:rightmost-lift-principal-block}
		Corollary~\ref{cor:recursive-principal-block} gives the principal block corresponding to the hull-dimension formula. The lower rows and columns are the rightmost lifts of the lower principal block. Therefore they are not the simple uniform translate $z_r^{\Qm}E_{r-2}(u)$; they form the principal block produced by the rightmost lift, which gives the same evaluation as the lower interval and is automatically disjoint from the top layer.
	\end{remark}
	
	\begin{corollary}\label{cor:algorithmic}
		Let $v\in I_r$. Repeated application of Theorem~\ref{thm:interval-recursion} computes $\Delta_r(v)$ after $\lfloor r/2\rfloor$ descents to the base intervals $I_0$ and $I_1$. Moreover, Corollary~\ref{cor:recursive-principal-block} builds a full-rank principal block $E_r(v)$.
	\end{corollary}
	
	\begin{proof}
		Each application of Theorem~\ref{thm:interval-recursion} lowers the interval index by $2$, so after $\lfloor r/2\rfloor$ steps one reaches either $I_0$ or $I_1$. Corollary~\ref{cor:recursive-principal-block} applies this descent to the principal blocks.
	\end{proof}
	
	\section{Explicit interval formulas and complete hull dimensions}
	
	The recursion ends at the base intervals recorded in Corollaries~\ref{cor:prototype-I0} and \ref{cor:prototype-I1}:
	\begin{equation}\label{eq:base-intervals}
		\Delta_0(v)=1\qquad (0<v<\Qm/2),
	\end{equation}
	\begin{equation}\label{eq:base-intervals-odd}
		\Delta_1(v)=2v-\Qm+1=A_1(v)\qquad (\Qm/2<v<\Qm).
	\end{equation}
	For convenience, put $A_0(v):=1$.
	
	\begin{corollary}\label{cor:closed-forms}
		Let $\ell\ge 1$.
		
		For $v\in I_{2\ell}$,
		\[
		\Delta_{2\ell}(v)
		=A_{2\ell}(v)+A_{2\ell-2}(v-\Qm)+\cdots+A_2\bigl(v-(\ell-1)\Qm\bigr)+1.
		\]
		For $v\in I_{2\ell+1}$,
		\[
		\Delta_{2\ell+1}(v)
		=A_{2\ell+1}(v)+A_{2\ell-1}(v-\Qm)+\cdots+A_1\bigl(v-\ell\Qm\bigr).
		\]
	\end{corollary}
	
	\begin{proof}
		We iterate Theorem~\ref{thm:interval-recursion} along the even and odd chains and terminate with \eqref{eq:base-intervals} and \eqref{eq:base-intervals-odd}.
	\end{proof}
	
	\begin{corollary}\label{cor:explicit-A-substitution}
		Let $\ell\ge 1$.
		
		For $v\in I_{2\ell}$,
		\[
		\Delta_{2\ell}(v)
		=
		1+\sum_{s=1}^{\ell}
		\sum_{t=(s+\ell)\Qm-v}^{\,v-(\ell-s)\Qm}
		\sum_{j=0}^{2s}
		(-1)^j\binom{2s}{j}\binom{t-jq+2s-1}{2s-1}.
		\]
		
		For $v\in I_{2\ell+1}$,
		\[
		\Delta_{2\ell+1}(v)
		=
		\sum_{s=0}^{\ell}
		\sum_{t=(s+\ell+1)\Qm-v}^{\,v-(\ell-s)\Qm}
		\sum_{j=0}^{2s+1}
		(-1)^j\binom{2s+1}{j}\binom{t-jq+2s}{2s}.
		\]
	\end{corollary}
	
	\begin{proof}
		We give the substitution in Corollary~\ref{cor:closed-forms} term by term.
		
		First suppose that $v\in I_{2\ell}$. For $1\le s\le \ell$, set
		\[
		w_s:=v-(\ell-s)\Qm.
		\]
		Then $w_s\in I_{2s}$, and Corollary~\ref{cor:closed-forms} becomes
		\[
		\Delta_{2\ell}(v)=1+\sum_{s=1}^{\ell}A_{2s}(w_s).
		\]
		Proposition~\ref{prop:A-explicit} gives
		\[
		A_{2s}(w_s)
		=
		\sum_{t=2s\Qm-w_s}^{w_s}
		\sum_{j=0}^{2s}
		(-1)^j\binom{2s}{j}\binom{t-jq+2s-1}{2s-1}.
		\]
		Since
		\[
		2s\Qm-w_s
		=
		2s\Qm-\bigl(v-(\ell-s)\Qm\bigr)
		=
		(s+\ell)\Qm-v,
		\]
		and
		\[
		w_s=v-(\ell-s)\Qm,
		\]
		we get the displayed formula for $\Delta_{2\ell}(v)$.
		
		Suppose now that $v\in I_{2\ell+1}$. For $0\le s\le \ell$, set
		\[
		w_s:=v-(\ell-s)\Qm.
		\]
		Then $w_s\in I_{2s+1}$, and Corollary~\ref{cor:closed-forms} becomes
		\[
		\Delta_{2\ell+1}(v)=\sum_{s=0}^{\ell}A_{2s+1}(w_s).
		\]
		Proposition~\ref{prop:A-explicit} yields
		\[
		A_{2s+1}(w_s)
		=
		\sum_{t=(2s+1)\Qm-w_s}^{w_s}
		\sum_{j=0}^{2s+1}
		(-1)^j\binom{2s+1}{j}\binom{t-jq+2s}{2s}.
		\]
		Using
		\[
		(2s+1)\Qm-w_s
		=
		(2s+1)\Qm-\bigl(v-(\ell-s)\Qm\bigr)
		=
		(s+\ell+1)\Qm-v,
		\]
		and
		\[
		w_s=v-(\ell-s)\Qm,
		\]
		we get the displayed formula for $\Delta_{2\ell+1}(v)$.
	\end{proof}

	We state the lower-half and full hull-dimension consequences. Let $k_{q,m}(v)$ denote the known dimension of $\PRM(q,m,v)$.
	
	\begin{corollary}\label{cor:lower-half}
		Let $m\ge 1$ and $0<v<m\Qm/2$ with $2v\not\equiv 0\pmod{\Qm}$. Choose the unique $r$ such that $v\in I_r$. Then automatically $m\ge r+1$, and
		\[
		\dim \Hull\bigl(\PRM(q,m,v)\bigr)=k_{q,m}(v)-\Delta_r(v),
		\]
		where $\Delta_r(v)$ is determined from Corollary~\ref{cor:closed-forms}.
	\end{corollary}
	
	\begin{proof}
		This follows from the hull-rank identity and Corollary~\ref{cor:closed-forms}.
	\end{proof}
	
	To pass from the lower intervals to the dual parameter intervals above the midpoint, we use the known duality results. If $v\not\equiv 0\pmod{\Qm}$, then the dual of $\PRM(q,m,v)$ is $\PRM(q,m,m\Qm-v)$. If $v\equiv 0\pmod{\Qm}$ and $m\Qm/2<v<m\Qm$, Song and Luo proved that
	\[
	\Hull\bigl(\PRM(q,m,v)\bigr)=\PRM(q,m,m\Qm-v).
	\]
	The boundary case $2v\equiv 0\pmod{\Qm}$ in the lower half is self-orthogonal by Kaplan and Kim \cite{KaplanKim}.
	
	\begin{corollary}\label{cor:complete}
		Let $v\ge 1$.
		\begin{enumerate}[label=\textup{(\roman*)},leftmargin=1.7em]
			\item When $0<v\le m\Qm/2$ and $2v\equiv 0\pmod{\Qm}$, $\PRM(q,m,v)$ is self-orthogonal, so
			\[
			\dim \Hull\bigl(\PRM(q,m,v)\bigr)=k_{q,m}(v).
			\]
			\item When $0<v<m\Qm/2$ and $2v\not\equiv 0\pmod{\Qm}$, choose $r$ by $v\in I_r$. Then $m\ge r+1$, and Corollary~\ref{cor:lower-half} applies:
			\[
			\dim \Hull\bigl(\PRM(q,m,v)\bigr)=k_{q,m}(v)-\Delta_r(v),
			\qquad v\in I_r.
			\]
			\item When $m\Qm/2<v<m\Qm$ and $v\not\equiv 0\pmod{\Qm}$, let $\mu:=m\Qm-v$. Then
			\[
			\dim \Hull\bigl(\PRM(q,m,v)\bigr)=\dim \Hull\bigl(\PRM(q,m,\mu)\bigr),
			\]
			so the dimension is obtained from items \textup{(i)} and \textup{(ii)} at the complementary parameter $\mu$.
			\item When $m\Qm/2<v<m\Qm$ and $v\equiv 0\pmod{\Qm}$,
			\[
			\Hull\bigl(\PRM(q,m,v)\bigr)=\PRM(q,m,m\Qm-v),
			\]
			so
			\[
			\dim \Hull\bigl(\PRM(q,m,v)\bigr)=k_{q,m}(m\Qm-v).
			\]
			\item When $v=m\Qm$, $\PRM(q,m,v)$ is LCD, so the hull dimension is $0$.
			\item When $v\ge m\Qm+1$, $\PRM(q,m,v)=\Fq^{|\PP^m(\Fq)|}$ by S\o rensen's result \cite[Remark~3]{Sorensen}, so the hull dimension is $0$.
		\end{enumerate}
		It follows that the hull dimension is determined for every degree $v\ge 0$ once one also includes the trivial case $v=0$.
	\end{corollary}
	
	\begin{proof}
		Item \textup{(i)} is the self-orthogonal part of Kaplan and Kim. Item \textup{(ii)} is Corollary~\ref{cor:lower-half}. Item \textup{(iii)} follows from the known dual identity when $v\not\equiv 0\pmod{\Qm}$ and the fact that $\Hull(C)=\Hull(C^\perp)$. Item \textup{(iv)} is the theorem of Song and Luo for the range above the midpoint with $v\equiv 0\pmod{\Qm}$. Item \textup{(v)} is the endpoint LCD case proved by Kaplan and Kim \cite[Remark~3.11]{KaplanKim}. Item \textup{(vi)} is the known statement that projective Reed-Muller codes equal the whole space in degrees $v\ge m(q-1)+1$; see \cite[Remark~3]{Sorensen}.
	\end{proof}
	
\begin{remark}\label{rem:boundary-cases}
	The case $v=0$ is direct. Here $\PRM(q,m,0)=\Span\{1\}$ and its hull is zero since $\langle 1,1\rangle=|\PP^m(\Fq)|\neq0$ in $\Fq$. If $m\Qm$ is even and $v=m\Qm/2$, then the degree lies on the self-orthogonal boundary and is already covered by item \textup{(i)}. At the upper end, the full-space case begins only at $v\ge m\Qm+1$; the endpoint $v=m\Qm$ is a separate LCD case already covered by item \textup{(v)}.
\end{remark}

\section{Conclusion}
	
	The main recursive step is the passage from $I_{r-2}$ to $I_r$ in the open lower-half intervals. Theorem~\ref{thm:interval-recursion} provides this step, and together with the known duality and boundary cases it yields the hull-dimension formula for projective Reed-Muller codes in all degrees.
	
	After the numerical formula, Corollary~\ref{cor:recursive-principal-block} gives recursive principal blocks and Corollary~\ref{cor:algorithmic} gives the computation of the defect together with a full-rank principal block. Therefore the recursion determines both the hull dimension and a corresponding principal block.

\end{document}